\documentclass[11pt]{article}

\usepackage[T1]{fontenc}
\usepackage{lmodern}
\usepackage[margin=1in]{geometry}
\usepackage{amsmath,amssymb,amsthm,mathtools}
\usepackage{nicefrac}
\usepackage{microtype}
\usepackage{needspace}
\usepackage{enumitem}
\usepackage[hidelinks]{hyperref}

\newtheorem{theorem}{Theorem}
\newtheorem{lemma}[theorem]{Lemma}

\theoremstyle{definition}
\newtheorem{definition}[theorem]{Definition}
\newtheorem{algorithm}[theorem]{Algorithm}
\theoremstyle{remark}

\newcommand{\OPT}{\operatorname{OPT}}
\newcommand{\E}{\mathbb{E}}
\newcommand{\Prb}{\mathbb{P}}

\title{Vertex Cover Problem with Advice}
\author{
  Lehilton Lelis Chaves Pedrosa \& Lucas de Oliveira Silva\\[0.5em]
  State University of Campinas (Unicamp) - Brazil\\
  \texttt{\{lehilton,lucas.oliveira.silva\}@ic.unicamp.br}
}
\date{}
\hypersetup{pdftitle={Vertex Cover Problem with Advice}}

\begin{document}

\maketitle

\begin{abstract}
  We study the minimum Vertex Cover problem under two non-adaptive offline advice models, where predictions about a fixed optimum solution are provided once as part of the input.
  These forms of prediction were introduced independently by Cohen-Addad, d'Orsi, Gupta, Lee, and Panigrahi (2024) and by Ghoshal, Makarychev, and Makarychev (2025).

  In Partial Predictions, independently revealed vertices come with correct membership labels.
  For every sufficiently small fixed $\varepsilon>0$, we give a randomized polynomial-time algorithm that has expected approximation ratio at most
  \[
    2-(2-o(1))\frac{\log\log(\nicefrac{1}{\varepsilon})}{\log(\nicefrac{1}{\varepsilon})}.
  \]

  In Noisy Predictions, every vertex instead receives a mutually independent noisy membership label of bias $\varepsilon$.
  For every sufficiently small fixed $\varepsilon>0$, we give a randomized polynomial-time algorithm that has expected approximation ratio at most
  \[
    2-(1-o(1))\frac{\log\log(\nicefrac{1}{\varepsilon})}{\log(\nicefrac{1}{\varepsilon})}.
  \]
  Our Noisy Predictions bound matches the leading asymptotic improvement below $2$ obtained by Aamand, Chen, Gollapudi, Silwal, and Wu (2025), despite using only one noisy label per vertex rather than a separate independent label for each endpoint of each incident edge.
  Both our bounds hold as $\varepsilon$ goes to $0$ and are strictly below $2$, the optimal approximation threshold for minimum Vertex Cover without advice under the Unique Games Conjecture, as shown by Khot and Regev (2008).
\end{abstract}

\section{Introduction}

A growing body of research has explored when \emph{learning-augmented} algorithms, which receive additional predictions from (pre-trained) machine-learning models, can overcome worst-case barriers for classical hard computational problems.
Cohen-Addad~et~al.~\cite{CdGLP24} studied Max Cut under \emph{Partial Predictions}, which reveal correct labels for randomly selected vertices, and \emph{Noisy Predictions}, which provide one noisy label per vertex; they improved the Max Cut guarantee within these models and also studied Max-2-CSPs.
Independently and concurrently, Ghoshal, Makarychev, and Makarychev~\cite{GMM25} studied corresponding advice models for CSPs, formulated with mutually independent advice coordinates.

We consider the minimum Vertex Cover problem, which seeks the smallest set of vertices that covers every edge.
This problem admits a polynomial-time $2$-approximation, while under the Unique Games Conjecture, any fixed factor below $2$ is NP-hard~\cite{KhotRegev08}.
We investigate how predictions about membership in a fixed minimum cover can improve the approximation guarantee.

For every sufficiently small fixed advice parameter $\varepsilon>0$, Theorems~\ref{thm:main-upper} and~\ref{thm:noisy-upper} give randomized polynomial-time algorithms that always return a vertex cover and have expected approximation ratios at most
\[
  2-(2-o(1))\frac{\log\log(\nicefrac{1}{\varepsilon})}{\log(\nicefrac{1}{\varepsilon})}
  \qquad\text{and}\qquad
  2-(1-o(1))\frac{\log\log(\nicefrac{1}{\varepsilon})}{\log(\nicefrac{1}{\varepsilon})},
\]
for Partial Predictions and Noisy Predictions, respectively, as $\varepsilon\to0$.
Here $\varepsilon$ is the reveal probability or noisy-label bias, respectively, under the conventions formalized below, and the expectations include both advice and algorithmic randomness.
Our noisy guarantee matches the leading asymptotic improvement below $2$ implied by the analysis of Aamand, Chen, Gollapudi, Silwal, and Wu~\cite{ACGSW25}, after rescaling the bias parameter by a constant.
However, our model is arguably weaker.
Indeed, while their model supplies an independent prediction for each edge endpoint, so a degree-$d$ vertex receives $d$ predictions about its membership, our model supplies one prediction per vertex.

Both algorithms preprocess the graph to obtain a bounded maximum degree, then combine the already selected vertices with a cover of the residual graph computed by Halperin's approximation algorithm~\cite{Halperin02}.
With partial predictions, we first select revealed cover vertices and all neighbors of revealed vertices outside the cover.
Hence, every vertex selected at this stage belongs to the fixed optimum.
We then select all vertices above a degree threshold in the residual graph, bounding the expected number selected outside the optimum using the certified selections.
With noisy predictions, we process decreasing degree thresholds and use randomized signs derived from neighbors' predictions to decide whether to select a high-degree vertex or its current neighbors.
Each residual graph depends on signs used in earlier iterations, so independence of the original labels alone does not justify concentration on the surviving vertices.
However, our construction makes the signs for the current iteration independent across vertices conditional on the earlier iteration signs, with controlled bias toward their true labels.
Both analyses bound the cost of preprocessing relative to the number of optimum vertices it removes, preserving the leading improvement from the bounded-degree approximation.
The respective final degree thresholds $\Theta(\nicefrac{1}{\varepsilon})$ and $\Theta(\nicefrac{\log(\nicefrac{1}{\varepsilon})}{\varepsilon^2})$ yield the leading coefficients $2$ and $1$ in our guarantees when substituted into Halperin's bound.

\subsection{Models and notation}

Throughout, $\log$ denotes the base-two logarithm.

For a graph $G=(V, E)$, let $\OPT(G)$ be the size of a minimum vertex cover, let $d_G(v)$ be the degree of $v$ in $G$, and let $\Delta(G)$ be the maximum degree.
Given $U\subseteq V(G)$, write $G-U$ for the subgraph induced by $V(G)\setminus U$.
For $A, B\subseteq V(G)$, write $E_G(A, B)$ for the set of edges with one endpoint in $A$ and the other in $B$.
For $U\subseteq V(G)$, let
\[
  N_G(U):=\{w\in V(G)\setminus U\mid\text{$w$ has a neighbor in $U$}\},
\]
and write $N_G(v):=N_G(\{v\})$ for the \emph{open neighborhood} of $v$ and $N_G[v]:=N_G(v)\cup\{v\}$ for its \emph{closed neighborhood}.

Let $G=(V, E)$ be a graph and put $n:=|V|$.
Fix a minimum vertex cover $C^*\subseteq V$ and set
\[
  k^*:=|C^*|=\OPT(G), \qquad\text{and}\qquad I^*:=V\setminus C^*.
\]
The set $C^*$ is used only to define the ground-truth advice and in the analysis; an algorithm does not receive $C^*$.
We assume $k^*\ge1$; if $k^*=0$, then $G$ has no edges and the empty set is trivially a minimum vertex cover.
Since $C^*$ is a vertex cover, $I^*$ is independent.
Let $x^*:=(x_v^*)_{v\in V}\in\{\pm1\}^V$ be the signed membership vector of $C^*$, with $x_v^*=1$ for $v\in C^*$ and $x_v^*=-1$ for $v\in I^*$.
For all algorithms below, write $\widehat C$ for its final output.

\begin{definition}[Partial Predictions]
  \label{def:subset-advice}
  Fix $\varepsilon\in(0,1]$.
  The algorithm receives a random set $S\subseteq V$ in which each vertex is revealed independently with probability $\varepsilon$, together with the correct labels $(x_v^*)_{v\in S}$.
\end{definition}

\begin{definition}[Noisy Predictions]
  \label{def:label-advice}
  Fix $\varepsilon\in(0,1]$.
  The algorithm receives an advice vector $\widetilde x:=(\widetilde x_v)_{v\in V}\in\{\pm1\}^V$, whose coordinates are mutually independent and satisfy
  \[
    \Prb[\widetilde x_v=x_v^*]=\frac{1+\varepsilon}{2},
    \qquad\text{and}\qquad
    \Prb[\widetilde x_v=-x_v^*]=\frac{1-\varepsilon}{2}.
  \]
\end{definition}

The algorithms below form families indexed by sufficiently small fixed advice parameter $\varepsilon>0$.
For simplicity, we state each family as a single algorithm with access to $\varepsilon$.
Its running time is polynomial in $n$, with constants and exponents that may depend on $\varepsilon$.
The approximation guarantees take the limit $\varepsilon\to0$ and hold for every input size $n$.

\paragraph{Independence convention.}

The independence assertions in Definitions~\ref{def:subset-advice} and~\ref{def:label-advice} are conditional on the fixed ground-truth vector $x^*$.
The algorithm's internal random coins are mutually independent and independent of the advice.
Random variables computed from the advice and internal coins may be dependent even though the underlying advice coordinates and coins are independent.

Both guarantees hold for every fixed choice of the ground-truth minimum cover $C^*$.
Unless stated otherwise, expectations include all advice randomness and all internal randomness of the algorithm under discussion.

\section{General Framework}

We use Halperin's bounded-degree approximation as a black box \cite{Halperin02}.
For all sufficiently large integers $\Delta$, there is a randomized polynomial-time algorithm $\mathcal B_\Delta$ and a number $\delta_\Delta\in(0,1)$ such that, whenever $\Delta(H)\le\Delta$, the algorithm returns a vertex cover of $H$ and
\[
  \E[|\mathcal B_\Delta(H)|]
  \le (2-\delta_\Delta)\cdot\OPT(H), \tag{1.1}\label{eq:halperin}
\]
where
\[
  \delta_\Delta
  =(2-o(1))\frac{\log\log\Delta}{\log\Delta}.
  \tag{1.2}\label{eq:halperin_delta}
\]
Here the $o(1)$ term tends to zero as $\Delta\to\infty$.
The expectation in \eqref{eq:halperin} is over the internal randomness of $\mathcal B_\Delta$.

\subsection{Progress--balance preprocessing}
\label{subsec:progress-balance}

Both of our algorithms for each type of advice begin with a \emph{preprocessing step}: they select vertices into their output and remove vertices from the graph, then invoke $\mathcal B_\Delta$ on the residual graph.
The following accounting statement uses the selected and removed vertices, together with a fixed optimum available only in the analysis.

\begin{lemma}
  \label{lem:master-accounting}
  Suppose a feasible preprocessing step permanently selects a set $A$ into the output, deletes vertices from $G$, and leaves a \emph{residual graph} $H$ with $\Delta(H)\le\Delta$.
  Assume that every edge removed is covered by $A$.
  Let $P^*$ be the number of vertices of the fixed cover $C^*$ deleted during the preprocessing step; we call $P^*$ the \emph{progress}.
  Set $B^*:=|A|-P^*$; we call $B^*$ the \emph{balance}.

  If the preprocessing step is followed by $\mathcal B_\Delta$ on $H$, then
  \[
    \E[|\widehat C|]
    \le
    (2-\delta_\Delta)\cdot k^*
    -(1-\delta_\Delta)\cdot\E[P^*]
    +\E[B^*].
    \tag{2}\label{eq:master-accounting}
  \]
\end{lemma}

\begin{proof}
  Exactly $P^*$ vertices of $C^*$ have been deleted, so $|C^*\cap V(H)|=k^*-P^*$.
  Every edge of $H$ is an edge of $G$ whose endpoints both survive.
  Since $C^*$ covers $G$, the set $C^*\cap V(H)$ covers $H$.

  Condition on the preprocessing history: the advice realization and all internal random choices made before the call to $\mathcal B_\Delta$.
  This fixes $A$, $H$, $P^*$, and $B^*$.
  Applying \eqref{eq:halperin} then gives
  \begin{align*}
  \E[|\widehat C|\mid\text{preprocessing}]
   & \le |A|+(2-\delta_\Delta)\cdot(k^*-P^*)                     \\
   & =P^*+B^*+(2-\delta_\Delta)\cdot(k^*-P^*)                    \\
   & =(2-\delta_\Delta)\cdot k^*-(1-\delta_\Delta)\cdot P^*+B^*.
  \end{align*}
  Taking the expectation of this conditional inequality over the preprocessing history averages $P^*$ and $B^*$ on its right-hand side, yielding \eqref{eq:master-accounting}.
\end{proof}

The quantities $P^*$ and $B^*$ are defined using the fixed cover $C^*$ and are not computed by either algorithm.
The progress counts vertices of $C^*$ removed by the preprocessing step, and the balance is the number selected into $A$ minus the progress.
In the preprocessing for Noisy Predictions, $B^*$ may be negative.
To use \eqref{eq:master-accounting}, each preprocessing step will bound its balance in the form
\[
  \E[B^*] < (1-\delta_\Delta)\cdot\E[P^*]+O(\alpha)k^*.
  \tag{3}\label{eq:amortization-template}
\]
Here $\alpha\ge0$ bounds the additive error per optimum vertex.
In \eqref{eq:master-accounting}, algorithm $\mathcal B_\Delta$ contributes $(2-\delta_\Delta)\cdot k^*$, saving $\delta_\Delta\cdot k^*$ relative to $2k^*$.
Requiring $\alpha=o(\delta_\Delta)$ makes the additive term $o(\delta_\Delta)k^*$, so it does not change the first-order saving.

\section{Partial Predictions}

The correct revealed labels let the algorithm select revealed cover vertices and the neighbors of revealed independent-set vertices.
Every vertex selected in this way belongs to the fixed optimum $C^*$.
It then selects and removes every vertex whose degree in that residual graph exceeds a prescribed threshold.
These selections and deletions form the preprocessing step analyzed by the framework from Section~\ref{subsec:progress-balance}.

\subsection{Preprocessing}

Partition $S$ according to its received advice signs: put vertices labeled $1$ in $R_{C^*}$ and vertices labeled $-1$ in $R_{I^*}$.
Because the revealed signs are correct, these computable sets satisfy
\[
  R_{C^*}:=S\cap C^*,
  \qquad\text{and}\qquad
  R_{I^*}:=S\cap I^*.
\]
Define the \emph{certified set} and the \emph{residual graph} by
\[
  S_0:=R_{C^*}\cup N_G(R_{I^*}),
  \qquad\text{and}\qquad
  G_0:=G-(S_0\cup R_{I^*}).
\]
Every vertex of $R_{I^*}$ is isolated in $G-S_0$, since $I^*$ is independent and all its neighbors lie in $N_G(R_{I^*})\subseteq S_0$.
Thus, removing $R_{I^*}$ as well does not remove any additional edge.
\begin{lemma}
  \label{lem:certified}
  For every realization of the advice, $S_0\subseteq C^*$ and $\OPT(G_0)=k^*-|S_0|=|C^*\setminus S_0|$.
\end{lemma}

\begin{proof}
  Since $R_{C^*}\subseteq C^*$ and $R_{I^*}\subseteq I^*$, and $I^*$ is independent, every neighbor of a vertex in $R_{I^*}$ belongs to $C^*$.
  Thus, $S_0\subseteq C^*$.

  The set $C^*\setminus S_0$ covers $G_0$, so $\OPT(G_0)\le k^*-|S_0|$.
  Conversely, if $C$ covers $G_0$, then every edge of $G$ either meets $S_0$ or survives in $G_0$: any edge incident to $R_{I^*}$ has its other endpoint in $N_G(R_{I^*})\subseteq S_0$.
  Hence, $C\cup S_0$ covers $G$.
  Therefore, $k^*\le |C\cup S_0|\le |C|+|S_0|$, so $|C|\ge k^*-|S_0|$.
  Minimizing over covers $C$ of $G_0$ proves the reverse inequality.
\end{proof}

\subsection{High-degree removal}

Fix an integer threshold $\Delta$ large enough that $\mathcal B_\Delta$ is available.
The first three steps below construct the selected vertices and residual graph required by the framework from Section~\ref{subsec:progress-balance}; the final step performs the completion prescribed there.

\begin{algorithm}
  \label{alg:threshold}
  Given $G$ and the partial predictions:
  \begin{enumerate}[label=\textnormal{(\arabic*)},leftmargin=2.2em,itemsep=2pt,topsep=2pt]
    \item construct $S_0$ and $G_0=G-(S_0\cup R_{I^*})$;
    \item set $S_1:=\{u\in V(G_0)\mid d_{G_0}(u)>\Delta\}$;
    \item let $G_1:=G_0-S_1$;
    \item complete the framework by returning $\widehat C:=S_0\cup S_1\cup\mathcal B_\Delta(G_1)$.
  \end{enumerate}
\end{algorithm}

For the analysis, let
\[
  a_v^*:=|N_G(v)\cap I^*|,
  \qquad\text{and}\qquad
  r_v^*:=\Prb[v\notin S_0].
\]
We call $a_v^*$ the \emph{exposure} of $v$ to $I^*$, and $r_v^*$ is its \emph{survival probability} after the certified preprocessing.

Also let
\[
  T^*:=S_1\cap C^*,
  \qquad\text{and}\qquad
  F^*:=S_1\cap I^*.
\]
Thus, $T^*$ is the set of \emph{true positives}, namely the vertices placed in $S_1$ that belong to $C^*$.
The set $F^*$ consists of the \emph{false positives}, namely the vertices placed in $S_1$ that belong to $I^*$.
Finally, set
\[
  M^*:=|E_G(C^*\setminus S_0,I^*)|.
\]

\begin{lemma}
  \label{lem:feasible}
  Algorithm~\ref{alg:threshold} always returns a vertex cover of $G$, and $\Delta(G_1)\le\Delta$.
\end{lemma}

\begin{proof}
  Every vertex retained in $G_1$ has degree at most $\Delta$ in $G_0$, and vertex deletion cannot increase degree.
  Every edge of $G$ is covered at the first stage at which one of its endpoints is deleted.
  Hence, an edge surviving both deletions lies in $G_1$ and is covered by $\mathcal B_\Delta(G_1)$.
\end{proof}

\begin{lemma}
  \label{lem:survival}
  For every $v\in C^*$,
  \[
    a_v^*\cdot r_v^*\le\frac{1-\varepsilon}{\varepsilon}(1-r_v^*).
  \]
\end{lemma}

\begin{proof}
  Let
  \[
    Y_v^*:=\sum_{w\in\{v\}\cup(N_G(v)\cap I^*)}
    \mathbf 1_{\{w\in S\}}.
  \]
  The event $v\notin S_0$ is exactly $Y_v^*=0$.
  Independence gives
  \[
    \E[Y_v^*]=\varepsilon(a_v^*+1),
    \qquad\text{and}\qquad
    \operatorname{Var}(Y_v^*)=\varepsilon(1-\varepsilon)(a_v^*+1).
  \]
  By Cauchy--Schwarz,
  \begin{align*}
  \E[Y_v^*]^2
   & =\E[Y_v^* \cdot \mathbf 1_{\{Y_v^*>0\}}]^2         \\
   & \le\E[(Y_v^*)^2]\cdot\E[\mathbf 1_{\{Y_v^*>0\}}^2] \\
   & =\E[(Y_v^*)^2]\cdot\Prb[Y_v^*>0].
  \end{align*}
  Since $r_v^*=\Prb[Y_v^*=0]$ and $\E[(Y_v^*)^2]=\E[Y_v^*]^2+\operatorname{Var}(Y_v^*)$, rearranging gives
  \[
    r_v^*\cdot\E[Y_v^*]^2
    \le(1-r_v^*)\cdot\operatorname{Var}(Y_v^*).
  \]
  Multiplying by $a_v^*/\E[Y_v^*]^2$ and substituting the values for the mean and variance gives
  \[
    a_v^*\cdot r_v^*
    \le(1-r_v^*)\cdot\frac{a_v^*\cdot(1-\varepsilon)}
    {\varepsilon(a_v^*+1)}
    \le\frac{1-\varepsilon}{\varepsilon}(1-r_v^*).
  \]
\end{proof}

Selecting vertices of $S_1$ outside $C^*$ increases the output size without removing a vertex of $C^*$.
To bound their number, we use $M^*$.
Although the preprocessing also removes $R_{I^*}$, each such vertex has all its neighbors in $S_0$.
Thus, this removal deletes no additional $C^*$--$I^*$ edge, and $M^*$ is exactly the number of residual $C^*$--$I^*$ edges in $G_0$.

\begin{lemma}
  \label{lem:cross}
  For every realization of the advice, $M^*=\sum_{v\in C^*}a_v^*\cdot\mathbf 1_{\{v\notin S_0\}}$.
  Hence, $$\E[M^*]=\sum_{v\in C^*}a_v^*\cdot r_v^*.
  $$
\end{lemma}

\begin{proof}
  For each $v\in C^*$, all $a_v^*$ corresponding cross edges survive in $G_0$ if $v\notin S_0$: none of its neighbors in $I^*$ can then lie in $R_{I^*}$, since that would place $v$ in $N_G(R_{I^*})\subseteq S_0$.
  If $v\in S_0$, none of those edges survives.
  Summing over $C^*$ gives the deterministic identity; linearity of expectation gives the displayed expectation.
\end{proof}

\begin{lemma}
  \label{lem:falsepositive}
  Deterministically, $ |F^*|\le\frac{M^*}{\Delta+1}.
  $ Consequently, $ \E[|F^*|] \le\frac1{\Delta+1}\sum_{v\in C^*}a_v^*\cdot r_v^*. $
\end{lemma}

\begin{proof}
  Every neighbor in $G_0$ of a vertex in $I^*$ lies in $C^*\setminus S_0$.
  For $u\in F^*$, integrality and $d_{G_0}(u)>\Delta$ give $d_{G_0}(u)\ge\Delta+1$.
  Moreover, every residual $C^*$--$I^*$ edge is counted for at most one endpoint in $F^*$.
  Hence,
  \[
    (\Delta+1)|F^*|
    \le\sum_{u\in F^*}d_{G_0}(u)
    \le M^*.
  \]
  Taking expectations and using Lemma~\ref{lem:cross} gives the expectation bound.
\end{proof}

\subsection{Approximation guarantee}

For steps (1)--(3) of Algorithm~\ref{alg:threshold}, instantiate the general \emph{progress} and \emph{balance} from Section~\ref{subsec:progress-balance} as:
\[
  P^*:=|S_0|+|T^*|,
  \qquad
  A:=S_0\cup S_1,
  \qquad\text{and}\qquad
  B^*:=|A|-P^*=|F^*|.
  \tag{4}\label{eq:subset-progress-balance}
\]
Thus, in this case, the balance is exactly the false-positive count $|F^*|$.
The sets $S_0$ and $S_1$ are disjoint.
Since both are selected into the output before being deleted, every edge removed by this preprocessing is covered by $A$.

\begin{lemma}
  \label{lem:subset-amortized}
  If steps (1)--(3) of Algorithm~\ref{alg:threshold} are performed with any integer $\Delta\ge0$, then they satisfy
  \[
    \E[M^*]\le\frac{1-\varepsilon}{\varepsilon}\E[|S_0|].
  \]
  Consequently,
  \[
    \E[B^*]
    \le
    \frac{1-\varepsilon}{\varepsilon(\Delta+1)}\E[P^*].
    \tag{5}\label{eq:subset-amortized}
  \]
\end{lemma}

\begin{proof}
  Since $S_0\subseteq C^*$, linearity of expectation gives $\E[|S_0|]=\sum_{v\in C^*}(1-r_v^*)$.
  Lemmas~\ref{lem:survival} and~\ref{lem:cross} therefore imply
  \[
    \E[M^*]=\sum_{v\in C^*}a_v^*\cdot r_v^*
    \le\frac{1-\varepsilon}{\varepsilon}\sum_{v\in C^*}(1-r_v^*)
    =\frac{1-\varepsilon}{\varepsilon}\E[|S_0|].
  \]
  Lemma~\ref{lem:falsepositive} and \eqref{eq:subset-progress-balance} now gives
  \[
    \E[B^*]=\E[|F^*|]
    \le\frac{\E[M^*]}{\Delta+1}
    \le\frac{(1-\varepsilon)\E[|S_0|]}{\varepsilon(\Delta+1)}
    \le\frac{(1-\varepsilon)\E[P^*]}{\varepsilon(\Delta+1)}.
  \]
\end{proof}

\begin{theorem}
  \label{thm:subset-finite}
  Let $\Delta$ be sufficiently large for \eqref{eq:halperin} to hold.
  Algorithm~\ref{alg:threshold} satisfies
  \[
    \E[|\widehat C|]
    \le
    (2-\delta_\Delta)\cdot k^*-
    \left(
    1-\delta_\Delta-
    \frac{1-\varepsilon}{\varepsilon(\Delta+1)}
    \right)\cdot\E[P^*].
    \tag{6}\label{eq:subset-finite}
  \]
  In particular, if
  \[
    \varepsilon(\Delta+1)(1-\delta_\Delta)\ge1-\varepsilon, \tag{7}\label{eq:subset-threshold-condition}
  \]
  then
  \[
    \E[|\widehat C|]
    \le(2-\delta_\Delta)\cdot\OPT(G).
  \]
\end{theorem}

\begin{proof}
  Apply Lemma~\ref{lem:master-accounting} to the preprocessing in Algorithm~\ref{alg:threshold}, then substitute \eqref{eq:subset-amortized}.
  Condition \eqref{eq:subset-threshold-condition} makes the coefficient of $\E[P^*]$ in \eqref{eq:subset-finite} nonnegative.
\end{proof}

\begin{theorem}
  \label{thm:main-upper}
  For every sufficiently small fixed $\varepsilon>0$, Algorithm~\ref{alg:threshold} with a suitable integer $\Delta$ is a randomized polynomial-time algorithm with partial predictions that always outputs a vertex cover and satisfies
  \[
    \frac{\E[|\widehat C|]}{\OPT(G)}
    \le
    2-(2-o(1))\frac{\log\log(\nicefrac{1}{\varepsilon})}{\log(\nicefrac{1}{\varepsilon})}
    \qquad (\varepsilon\to0).
  \]
\end{theorem}

\begin{proof}
  Fix any constant $\alpha>1$ and set
  \[
    \Delta_{\mathrm{sub}}
    :=\left\lceil\frac{\alpha(1-\varepsilon)}{\varepsilon}\right\rceil.
  \]
  Since $\delta_\Delta\to0$, for all sufficiently small $\varepsilon$ we have
  \[
    \varepsilon(\Delta_{\mathrm{sub}}+1)
    \cdot(1-\delta_{\Delta_{\mathrm{sub}}})
    \ge \alpha(1-\varepsilon)(1-\delta_{\Delta_{\mathrm{sub}}})
    \ge1-\varepsilon.
  \]
  Theorem~\ref{thm:subset-finite} therefore applies.
  Moreover,
  \[
    1-\varepsilon\longrightarrow1,
    \qquad\text{and}\qquad
    \Delta_{\mathrm{sub}}=\Theta(\nicefrac{1}{\varepsilon}).
  \]
  Writing $L:=\log(\nicefrac{1}{\varepsilon})$, we obtain
  \[
    \begin{aligned}
    \log\Delta_{\mathrm{sub}} & =L+O(1)=L\bigl(1+O(\nicefrac{1}{L})\bigr)    \\
    \text{and}\qquad\log\log\Delta_{\mathrm{sub}}
                              & =\log L+\log\bigl(1+O(\nicefrac{1}{L})\bigr)
    =\log L+O(\nicefrac{1}{L}).
    \end{aligned}
  \]
  Substitution in \eqref{eq:halperin_delta} gives
  \[
    \delta_{\Delta_{\mathrm{sub}}}
    =(2-o(1))\frac{\log L}{L}
    =(2-o(1))
    \frac{\log\log(\nicefrac{1}{\varepsilon})}{\log(\nicefrac{1}{\varepsilon})}.
  \]
  This proves the theorem.
\end{proof}

\section{Noisy Predictions}

Our algorithm for noisy predictions first uses each noisy prediction to generate signs for several iterations and then processes high-degree vertices under decreasing degree thresholds.
During an iteration, it uses the current neighbors' signs to decide whether to add each high-degree vertex or its current neighbors to the cover.
Each operation deletes the high-degree vertex, and when the algorithm adds its current neighbors, it deletes them as well.
Vertices that remain induce the subgraph for the following iteration.
These high-degree iterations instantiate the preprocessing framework from Section~\ref{subsec:progress-balance}, which we analyze using its progress and balance quantities.

The residual graph at one iteration depends on signs used in earlier iterations.
To apply concentration inequalities on that graph, the current signs must remain independent across vertices and biased toward their true labels after conditioning on those earlier signs.
The construction below is applied independently at each vertex; it places each noisy prediction in at most one coordinate, filling the others with independent fair signs.

\subsection{Information splitting}

For every integer $j\ge0$, set $c_j:=2^{-j-1}$ and $\eta_j:=\varepsilon \cdot c_j$.
Given an integer $J\ge0$, fix a vertex $v$ of true label $x_v^*$.
Let $\widetilde x_v\in\{\pm1\}$ be its noisy prediction, with $\Prb[\widetilde x_v=x_v^*]=(1+\varepsilon)/2$.
With no access to $x_v^*$, generate $\widetilde X_v=(\widetilde X_v^{(0)},\ldots,\widetilde X_v^{(J)})\in\{\pm1\}^{J+1}$ as follows.
Flip a fair coin until a head first appears, stopping after $J+1$ tails.
If a head first appears after $j$ tails, set $\widetilde X_v^{(j)}=\widetilde x_v$ and give every other iteration an independent fair sign.
If all $J+1$ flips are tails, give every iteration an independent fair sign.
The following lemma gives the probability of every complete vector of iteration signs.
It also gives the expected value of each sign and the covariance between signs from distinct iterations.

\begin{lemma}
  \label{lem:information-splitting}
  Fix the true sign $x_v^*\in\{\pm1\}$.
  For every $Y\in\{\pm1\}^{J+1}$, the vector produced by the above process satisfies
  \begin{equation}
  \tag{8}
  \label{eq:iteration-sign-law}
  \Prb[\widetilde X_v=Y]
  =2^{-(J+1)} \left(1+x_v^*\cdot\sum_{j=0}^J\eta_j\cdot Y^{(j)}\right).
  \end{equation}
  Moreover, $\E[\widetilde X_v^{(j)}]=\eta_j\cdot x_v^*$ for every $j$, and $\operatorname{Cov}(\widetilde X_v^{(i)},\widetilde X_v^{(j)})=-\eta_i\cdot\eta_j$ for distinct $i,j$.
\end{lemma}

\begin{proof}
  Fix $Y\in\{\pm1\}^{J+1}$.
  The probability that a head first appears after $j$ tails is $c_j$, and the probability that no head appears is $2^{-(J+1)}$.
  Conditional on the event that $\widetilde x_v$ is assigned after $j$ tails,
  \[
    \Prb[\widetilde X_v=Y\mid \widetilde x_v\text{ is assigned after $j$ tails}]
    =2^{-(J+1)} (1+\varepsilon \cdot x_v^*\cdot Y^{(j)}).
  \]
  On the other hand, conditional on no head appearing, $\widetilde X_v$ is uniform.
  Therefore, the law of total probability and $\sum_{j=0}^J c_j+2^{-(J+1)}=1$ give
  \begin{align*}
  \Prb[\widetilde X_v=Y]
   & =2^{-(J+1)} \left(\sum_{j=0}^Jc_j+x_v^*\cdot\sum_{j=0}^J\eta_j\cdot Y^{(j)}+2^{-(J+1)}\right) \\
   & =2^{-(J+1)} \left(1+x_v^*\cdot\sum_{j=0}^J\eta_j\cdot Y^{(j)}\right).
  \end{align*}
  Moreover, $\sum_{Z\in\{\pm1\}^{J+1}}Z^{(j)}=0$ for every $j$, and hence
  \[
    \sum_{Z\in\{\pm1\}^{J+1}}
    \Prb[\widetilde X_v=Z]
    =2^{-(J+1)}\cdot2^{J+1}=1.
  \]
  For every $j$,
  \[
    \E[\widetilde X_v^{(j)}]
    =\sum_{Z\in\{\pm1\}^{J+1}}Z^{(j)}\cdot\Prb[\widetilde X_v=Z].
  \]
  Substituting \eqref{eq:iteration-sign-law} into this expectation and expanding gives
  \begin{align*}
  \E[\widetilde X_v^{(j)}]
   & =2^{-(J+1)} \sum_{Z\in\{\pm1\}^{J+1}}Z^{(j)}\cdot
  \left(1+x_v^*\cdot\sum_{i=0}^J\eta_i\cdot Z^{(i)}\right) \\
   & =2^{-(J+1)} \sum_{Z\in\{\pm1\}^{J+1}}Z^{(j)}
  +2^{-(J+1)}\cdot x_v^*\cdot\sum_{i=0}^J\eta_i\cdot
  \sum_{Z\in\{\pm1\}^{J+1}}Z^{(j)}\cdot Z^{(i)}.
  \end{align*}
  The first sum in the expansion is zero because the positive and negative values of $Z^{(j)}$ cancel in pairs.
  In the second sum, $Z^{(j)}Z^{(i)}=1$ when $i=j$, and when $i\ne j$, the products cancel in pairs by changing the sign of coordinate $j$.
  Hence,
  \[
    2^{-(J+1)} \sum_{Z\in\{\pm1\}^{J+1}}Z^{(j)}\cdot Z^{(i)}
    =
    \begin{cases}
    1, & i=j,   \\
    0, & i\ne j
    \end{cases}
    .
  \]
  Therefore,
  \[
    \E[\widetilde X_v^{(j)}]
    =x_v^*\cdot\sum_{i=0}^J\eta_i\cdot\mathbf 1_{\{i=j\}}
    =\eta_j\cdot x_v^*.
  \]
  For distinct $i,j$, substituting \eqref{eq:iteration-sign-law} into the expectation of the product gives
  \begin{align*}
  \E[\widetilde X_v^{(i)}\cdot\widetilde X_v^{(j)}]
   & =\sum_{Z\in\{\pm1\}^{J+1}}
  Z^{(i)}Z^{(j)}\cdot\Prb[\widetilde X_v=Z]                       \\
   & =2^{-(J+1)} \sum_{Z\in\{\pm1\}^{J+1}}Z^{(i)}Z^{(j)}\cdot
  \left(1+x_v^*\cdot\sum_{k=0}^J\eta_k\cdot Z^{(k)}\right)        \\
   & =2^{-(J+1)} \sum_{Z\in\{\pm1\}^{J+1}}Z^{(i)}Z^{(j)}
  +2^{-(J+1)}\cdot x_v^*\cdot\sum_{k=0}^J\eta_k\cdot
  \sum_{Z\in\{\pm1\}^{J+1}}Z^{(i)}Z^{(j)}Z^{(k)}.
  \end{align*}
  The first sum is zero because flipping coordinate $i$ changes the sign of $Z^{(i)}Z^{(j)}$.
  For the second sum, when $k=i$, $(Z^{(i)})^2 = 1$, so the inner sum reduces to the sum of $Z^{(j)}$ over all sign vectors, which is zero.
  When $k=j$, it similarly reduces to the sum of $Z^{(i)}$, which is again zero.
  When $k\notin\{i,j\}$, pair each sign vector with the one obtained by flipping coordinate $k$.
  Their products $Z^{(i)}Z^{(j)}Z^{(k)}$ have opposite signs, so this inner sum is zero as well.
  So we conclude that
  \[
    \E[\widetilde X_v^{(i)}\cdot\widetilde X_v^{(j)}]=0.
  \]
  Since $(x_v^*)^2=1$, it follows that
  \[
    \operatorname{Cov}(\widetilde X_v^{(i)},\widetilde X_v^{(j)})
    =0-(\eta_i\cdot x_v^*)\cdot(\eta_j\cdot x_v^*)
    =-\eta_i\cdot\eta_j.
  \]
\end{proof}

\subsection{Iterated high-degree removal}

Now, let $\alpha:=\nicefrac{1}{32}$ and fix an integer $\Gamma>\nicefrac{256}{\alpha^2}$ and set
\[
  \Delta_{\mathrm{lab}}:=
  \left\lceil \Gamma\cdot\varepsilon^{-2}\cdot\log(\nicefrac{1}{\varepsilon})\right\rceil,
  \qquad\text{and}\qquad
  D_j:=8^j\cdot\Delta_{\mathrm{lab}}.
\]
Take $J$ to be the least nonnegative integer such that $D_J \ge \nicefrac{\Delta(G)}{8}$.
Since $\Delta(G)\le n-1$, we have $J=O_\varepsilon(\log n)$.
Choose a fixed rational $t_\varepsilon$ with $|t_\varepsilon-(1-\alpha)\varepsilon|\le\nicefrac{\alpha}{2} \cdot \varepsilon$, and set $\tau_j:=t_\varepsilon\cdot c_j$ for every integer $j \ge 0$.
Such a rational exists for every fixed $\varepsilon>0$ and belongs to the finite description of the corresponding algorithm.
Moreover, the choice of $t_\varepsilon$ gives
\[
  |\tau_j-(1-\alpha)\eta_j|\le\frac{\alpha}{2}\eta_j,
  \qquad\text{and}\qquad
  \left(1-\frac{3\alpha}{2}\right)\eta_j
  \le\tau_j\le
  \left(1-\frac{\alpha}{2}\right)\eta_j.
\]
With $J$ fixed, apply the construction in Lemma~\ref{lem:information-splitting} independently to each $\widetilde x_v$, obtaining iteration signs $\widetilde X_v$ with the distribution stated there.
The construction uses at most $2(J+1)$ fair coins and $O(J+1)$ work per vertex in the worst case.

These parameters and signs define the algorithm below.
The first two steps construct the selected vertices and residual graph required by the framework of Section~\ref{subsec:progress-balance}, and the final step completes the algorithm.

\begin{algorithm}
  \label{alg:noisy}
  Given $G$ and the iteration signs:
  \begin{enumerate}[label=\textnormal{(\arabic*)},leftmargin=2.2em,itemsep=2pt,topsep=2pt]
    \item set $H_J:=G$ and initialize $S$ as an empty vertex set;

    \item\label{step:noisy-iterations} for each $j=J,J-1,\ldots,0$ do

          \begin{enumerate}[label=\textnormal{(\alph*)},leftmargin=2em,itemsep=2pt,topsep=2pt]
            \item if $j<J$, let $H_j$ be the residual graph resulting from iteration $j+1$;
            \item let $W_j:=\{v\in V(H_j)\mid d_{H_j}(v)>D_j\}$;
            \item for every $v\in W_j$ compute
                  \[
                    T_v:=\sum_{u\in N_{H_j}(v)}
                    \widetilde X_u^{(j)},
                  \]
                  and then classify, for this iteration, $v$ as a \emph{pivot} if $T_v\ge\tau_j\cdot d_{H_j}(v)$, and as a \emph{take} otherwise, then freeze all classifications;
            \item process $W_j$ in an arbitrary fixed order: when $v$ is reached, do nothing if $v$ has been deleted or if its degree after the preceding deletions in this iteration is at most $D_j$; otherwise, if $v$ is a take, add $v$ to $S$ and delete $v$; if $v$ is a pivot, add its current open neighborhood to $S$ and delete its current closed neighborhood.
          \end{enumerate}

    \item complete the framework by returning $\widehat C:=S \cup\mathcal B_{\Delta_{\mathrm{lab}}}(H_0)$.
  \end{enumerate}
\end{algorithm}

The algorithm achieves the guarantee below.
Its proof uses two invariants: a degree invariant and a probabilistic invariant about the conditional distribution of the current iteration signs given the signs used in earlier iterations.

\begin{theorem}
  \label{thm:noisy-upper}
  For every sufficiently small fixed $\varepsilon>0$, Algorithm~\ref{alg:noisy} is a randomized polynomial-time algorithm with noisy predictions that always outputs a vertex cover and satisfies
  \[
    \frac{\E[|\widehat C|]}{\OPT(G)}
    \le
    2-(1-o(1))\frac{\log\log(\nicefrac{1}{\varepsilon})}{\log(\nicefrac{1}{\varepsilon})}
    \qquad (\varepsilon\to0).
  \]
\end{theorem}

\subsection{Required invariants}

We first verify that every iteration covers the edges it removes and leaves a degree bound for the next iteration.

\begin{lemma}
  \label{lem:noisy-invariant}
  At the start of iteration $j$, the residual graph satisfies $\Delta(H_j)\le8D_j$.
  After iteration $j$, its maximum degree is at most $D_j$.
  Every edge removed while processing a take or pivot node is covered by the vertices selected in that operation.
\end{lemma}

\begin{proof}
  The initial bound follows from the choice of $J$.
  During an iteration $j$, degrees only decrease.
  If a vertex ends the iteration with degree greater than $D_j$, then it belonged to $W_j$ at the start of the iteration, survived until it was processed, and still had degree greater than $D_j$ at that time; the algorithm would then have deleted it, a contradiction.
  Since $D_j=8D_{j-1}$, the invariant follows by induction.
  A take makes the algorithm select the deleted vertex itself.
  A pivot makes the algorithm select all of its current neighbors before deleting the pivot and those neighbors.
  Thus each operation covers all edges incident to vertices it removes.
\end{proof}

The degree invariant controls the residual graphs, while the next lemma supplies the conditional independence needed to analyze the classifications made on those graphs.

\begin{lemma}
  \label{lem:iteration-law}
  For $0\le k\le J+1$, let $\mathcal F_k$ contain all signs from iterations $i\ge k$; in particular, $\mathcal F_{J+1}$ contains no signs.
  Fix an iteration $j\in\{0,\ldots,J\}$.
  Conditional on $\mathcal F_{j+1}$, the variables $(\widetilde X_v^{(j)})_{v\in V}$ are independent.
  Moreover, for every $v\in V$, the conditional expectation of $\widetilde X_v^{(j)}$ is
  \[
    \E[\widetilde X_v^{(j)}\mid\mathcal F_{j+1}]
    =\frac{\eta_j\cdot x_v^*}
    {1+x_v^*\cdot\sum_{i=j+1}^J\eta_i\cdot\widetilde X_v^{(i)}}.
  \]
  Furthermore, the denominator lies in $(1-\varepsilon,1+\varepsilon)$.
\end{lemma}

\begin{proof}
  Fix a vertex $v\in V$ and the values of all signs in $\mathcal F_{j+1}$.
  Let $f_v=(f_v^{(i)})_{i=j+1}^J$ be their values at $v$, and put $\mathcal S_{f_v}:=\sum_{i=j+1}^J\eta_i\cdot f_v^{(i)}$.
  The event that the earlier signs are $f_v$ and the current sign is $z\in\{\pm1\}$ leaves precisely the coordinates $0\le i<j$ unrestricted.
  Let $Z=(Z^{(i)})_{i=0}^{j-1}\in\{\pm1\}^j$ range over assignments to these coordinates.
  For each $0\le i<j$, pairing assignments that differ only at coordinate $i$ gives
  \[
    \sum_{Z\in\{\pm1\}^j}Z^{(i)}=0.
  \]
  Summing the distribution in Lemma~\ref{lem:information-splitting} over $Z$ therefore yields
  \begin{align*}
    \Prb[(\widetilde X_v^{(i)})_{i=j+1}^J=f_v \land \widetilde X_v^{(j)}=z]
    &=\sum_{Z\in\{\pm1\}^j}\Prb[\widetilde X_v=(Z,z,f_v)] \\
    &=2^{-(J+1)}\sum_{Z\in\{\pm1\}^j}
      \left(1+x_v^*\cdot\left(\mathcal S_{f_v}+\eta_j\cdot z+\sum_{i=0}^{j-1}\eta_i\cdot Z^{(i)}\right)\right) \\
    &=2^{-(J-j+1)}(1+x_v^*\cdot(\mathcal S_{f_v}+\eta_j\cdot z)).
  \end{align*}
  Summing over $z\in\{\pm1\}$ gives the probability that the earlier signs are $f_v$:
  \[
    \Prb[(\widetilde X_v^{(i)})_{i=j+1}^J=f_v]
    =2^{-(J-j)} (1+x_v^*\cdot \mathcal S_{f_v}).
  \]
  Since $(\widetilde X_u)_{u\in V}$ are independent, conditioning on $\mathcal F_{j+1}$ leaves $(\widetilde X_u^{(j)})_{u\in V}$ independent and makes the conditional distribution at $v$ depend only on $f_v$.
  The definition of conditional probability now gives:
  \begin{align*}
    \Prb[\widetilde X_v^{(j)}=z\mid\mathcal F_{j+1}]
    &=\Prb[\widetilde X_v^{(j)}=z\mid f_v] \\
    &=\frac{1+x_v^*\cdot(\mathcal S_{f_v}+\eta_j\cdot z)}
      {2(1+x_v^*\cdot \mathcal S_{f_v})}.
  \end{align*}
  Since $\widetilde X_v^{(j)}\in\{\pm1\}$, subtracting the $z=-1$ case from the $z=1$ case gives the claimed conditional mean.
  Finally, because $|\mathcal S_{f_v}|\le\sum_{i=j+1}^J\eta_i<\varepsilon$, the denominator has the stated bounds.

\end{proof}

\subsection{Approximation guarantee}

We now combine the structural and probabilistic invariants to analyze the general progress and balance quantities from Section~\ref{subsec:progress-balance}.

\begin{lemma}
  \label{lem:noisy-amortized}
  Suppose $0<\varepsilon\le\alpha/(4+\alpha)$.
  For the iterations in step~\ref{step:noisy-iterations} of Algorithm~\ref{alg:noisy}, instantiate the general \emph{progress} and \emph{balance} from Section~\ref{subsec:progress-balance} by letting $P^*$ count the vertices of $C^*$ deleted during those iterations, taking $A$ to be the set of vertices added to $S$ during those iterations, and setting $B^*:=|A|-P^*$.

  The expected value of $B^*$ is at most
  \[
    \frac13\E[P^*]
    +O\!
    \left(
    \Delta_{\mathrm{lab}}\cdot
    e^{-\frac{\alpha^2}{128}\cdot\varepsilon^2\cdot\Delta_{\mathrm{lab}}}
    \right)k^*.
  \]
\end{lemma}

\begin{proof}
  \Needspace{8\baselineskip}
  Fix an iteration $j$ and condition on $\mathcal F_{j+1}$.
  The graph $H_j$ and the set $W_j$ are then fixed.
  Let $P_j^*$ and $B_j^*$ be the progress and balance contributions of iteration $j$.
  To analyze one operation, fix a vertex $v\in W_j$.
  Write $d:=d_{H_j}(v)$ and $a:=|N_{H_j}(v)\cap I^*|$ for its counts at the start of the iteration.
  If an operation processes $v$, let $d'$ be the degree of $v$ at that time.
  Let $P_v^*$ count the vertices of $C^*$ deleted by the operation, and let $B_v^*$ be the number of vertices added to $S$ minus $P_v^*$.
  For a pivot operation applied to a cover vertex $v$, let $a'$ be the number of neighbors of $v$ in $I^*$ when $v$ is processed.

  If $v\in I^*$, a pivot operation has $(P_v^*,B_v^*)=(d',0)$, while a take has $(P_v^*,B_v^*)=(0,1)$.
  If $v\in C^*$, a take operation has $(P_v^*,B_v^*)=(1,0)$, and a pivot has
  \[
    (P_v^*,B_v^*)=(d'+1-a',a'-1).
    \tag{9}\label{eq:noisy-local-accounting}
  \]
  Every vertex included in $S$ by an operation is deleted in that operation, and later operations act only on surviving vertices.
  Thus $P_j^*$ and $B_j^*$ sum these local contributions within iteration $j$, while $P^*=\sum_{j=0}^J P_j^*$ and $B^*=\sum_{j=0}^J B_j^*$.

  Consider a vertex of $C^*$ with $a<\alpha d$ that is classified as a pivot.
  By Lemma~\ref{lem:noisy-invariant}, $d\le8D_j$, so $a<8\alpha D_j$.
  When $v$ is processed, $d\ge d'>D_j>0$ and $a'\le a$, hence $a'<8\alpha d'$.
  Equation~\eqref{eq:noisy-local-accounting} gives
  \[
    \frac{B_v^*}{P_v^*}
    =\frac{a'-1}{d'+1-a'}
    <\frac{8\alpha d'}{d'-8\alpha d'}
    =\frac13.
    \tag{10}\label{eq:low-exposure-ratio}
  \]
  A take operation at a cover vertex and a pivot operation at an $I^*$-vertex have zero balance.
  For the remaining cases, let $E_{j,v}$ be the event that $v\in I^*$ is classified as a take, or that $v\in C^*$ with $a\ge\alpha d$ is classified as a pivot.
  These classifications are frozen before any vertex of $W_j$ is processed.
  If $E_{j,v}$ occurs and $v$ is not skipped, a take operation has balance one, while a pivot operation has balance $a'-1\le d'-1 < d\le8D_j$.
  Since $v$ was arbitrary, the total balance in the current iteration satisfies
  \[
    B_j^*
    \le\frac13P_j^*
    +8D_j\sum_{v\in W_j}\mathbf1_{E_{j,v}}.
  \]
  The sum may count vertices that are later skipped, which only increases the upper bound.

  We next bound $\Prb(E_{j,v}\mid\mathcal F_{j+1})$ for the fixed vertex $v$.
  Put $\mu_v:=\E[T_v\mid\mathcal F_{j+1}]$.
  The $a$ neighbors of $v$ in $I^*$ have true sign $-1$, and its other $d-a$ neighbors have true sign $+1$, so their true signs sum to $d-2a$.
  For each $u\in N_{H_j}(v)$, Lemma~\ref{lem:iteration-law} implies that the absolute difference between $\E[\widetilde X_u^{(j)}\mid\mathcal F_{j+1}]$ and $\eta_j x_u^*$ is less than $\eta_j \cdot \varepsilon/(1-\varepsilon)$.
  Linearity of expectation and the triangle inequality give:
  \[
    \left|\mu_v-\eta_j(d-2a)\right|
    \le\sum_{u\in N_{H_j}(v)}
      \left|\E[\widetilde X_u^{(j)}\mid\mathcal F_{j+1}]-\eta_j x_u^*\right|
    <\sum_{u\in N_{H_j}(v)}\eta_j\cdot\frac{\varepsilon}{1-\varepsilon}
    \le\frac{\alpha}{4}\eta_j \cdot d.
  \]

  If $v\in I^*$, then $a=0$ as $I^*$ is independent.
  The bound $\mu_v>\eta_j(d-2a)-\frac{\alpha}{4}\eta_j\cdot d$, together with $\tau_j\le(1-\alpha/2)\eta_j$, gives
  \[
    \mu_v-\tau_j \cdot d
    >\left(1-\frac{\alpha}{4}\right)\eta_j \cdot d-\tau_j \cdot d
    \ge\left(\frac{\alpha}{2}-\frac{\alpha}{4}\right)\eta_j \cdot d
    =\frac{\alpha}{4}\eta_j \cdot d.
  \]

  If $v\in C^*$ and $a\ge\alpha d$, then $d-2a\le(1-2\alpha)d$.
  The bound $\mu_v<\eta_j(d-2a)+\frac{\alpha}{4}\eta_j\cdot d$, together with $\tau_j\ge(1-3\alpha/2)\eta_j$, gives
  \[
    \tau_j \cdot d-\mu_v
    >\tau_j \cdot d-\left(1-2\alpha+\frac{\alpha}{4}\right)\eta_j \cdot d
    \ge\left(\frac{\alpha}{2}-\frac{\alpha}{4}\right)\eta_j \cdot d
    =\frac{\alpha}{4}\eta_j \cdot d.
  \]
  In the first case $E_{j,v}$ is the event $T_v<\tau_j \cdot d$; in the second it is the event $T_v\ge\tau_j \cdot d$.
  Conditional on $\mathcal F_{j+1}$, the neighborhood $N_{H_j}(v)$ is fixed, and Lemma~\ref{lem:iteration-law} makes its $d$ current signs independent.

  \begin{samepage}
  Since each summand of $T_v$ lies in $[-1,1]$, Hoeffding's inequality~\cite{Hoeffding63} bounds the conditional probability of either event by
  \[
    e^{-\frac{2(\nicefrac{\alpha}{4} \cdot \eta_j \cdot d)^2}{4d}}
    =e^{-\frac{\alpha^2}{32}\varepsilon^2c_j^2 \cdot d}
    <e^{-\frac{\alpha^2}{32}\varepsilon^2c_j^2 \cdot D_j}
    =e^{-\frac{\alpha^2}{128}\varepsilon^2 \cdot 2^j \cdot \Delta_{\mathrm{lab}}}.
  \]
  The final equality uses $c_j^2 \cdot D_j=2^{j-2} \cdot \Delta_{\mathrm{lab}}$.
  \end{samepage}

  The surviving vertices of $C^*$ cover $H_j$ and are at most $k^*$ in total.
  Lemma~\ref{lem:noisy-invariant} therefore gives $|E(H_j)|\le8D_j \cdot k^*$.
  Since every vertex in $W_j$ has degree greater than $D_j$, we have
  \[
    D_j \cdot |W_j|
    \le\sum_{v\in W_j}d_{H_j}(v)
    \le2|E(H_j)|
    \le16D_j \cdot k^*,
  \]
  Hence $|W_j|<16k^*$: the first inequality is strict if $W_j$ is nonempty, and otherwise the conclusion follows from $k^*\ge1$.
  Taking conditional expectations in the per-iteration balance bound now yields
  \[
    \begin{aligned}
      \E[B_j^*\mid\mathcal F_{j+1}]
      &\le\frac13\E[P_j^*\mid\mathcal F_{j+1}]
          +8D_j\sum_{v\in W_j}\Prb(E_{j,v}\mid\mathcal F_{j+1})\\
      &<\frac13\E[P_j^*\mid\mathcal F_{j+1}]
          +128D_j \cdot k^* \cdot e^{-\frac{\alpha^2}{128}\varepsilon^2 \cdot 2^j \cdot \Delta_{\mathrm{lab}}}.
    \end{aligned}
  \]
  Taking expectations over $\mathcal F_{j+1}$ and summing over $j=0,\ldots,J$ gives the following bound, since $P^*=\sum_{j=0}^J P_j^*$ and $B^*=\sum_{j=0}^J B_j^*$ and expectation is linear:
  \[
    \E[B^*]
    <\frac13\E[P^*]
       +128k^* \cdot \Delta_{\mathrm{lab}}
        \sum_{j=0}^J8^j \cdot e^{-\frac{\alpha^2}{128}\varepsilon^2 \cdot 2^j \cdot \Delta_{\mathrm{lab}}}.
  \]

  We bound the sum by extending it to every $j\ge0$.
  Because $\Gamma>256/\alpha^2$, $\varepsilon\le\alpha/(4+\alpha)$, and $\alpha=1/32$, we have
  \[
    \frac{\alpha^2}{128}\varepsilon^2 \cdot \Delta_{\mathrm{lab}}
    \ge\frac{\alpha^2 \cdot \Gamma}{128}\log(\nicefrac{1}{\varepsilon})
    >2\log\!\left(\frac{4+\alpha}{\alpha}\right)
    >14.
  \]
  Since $2^j\ge j + 1$ for $j\ge0$ and $8e^{-14}<\nicefrac12$, we obtain
  \[
    \begin{aligned}
      \sum_{j=0}^J8^j \cdot e^{-\frac{\alpha^2}{128}\varepsilon^2 \cdot 2^j \cdot \Delta_{\mathrm{lab}}}
      &<e^{-\frac{\alpha^2}{128}\varepsilon^2 \cdot \Delta_{\mathrm{lab}}}
        \sum_{j\ge0}\left(8e^{-\frac{\alpha^2}{128}\varepsilon^2 \cdot \Delta_{\mathrm{lab}}}\right)^j\\
      &<2e^{-\frac{\alpha^2}{128}\varepsilon^2 \cdot \Delta_{\mathrm{lab}}}.
    \end{aligned}
  \]
  Substitution finally gives $\E[B^*]<\frac13\E[P^*]+256\Delta_{\mathrm{lab}}\cdot \exp\!\left(-\frac{\alpha^2}{128} \cdot \varepsilon^2 \cdot \Delta_{\mathrm{lab}}\right) k^*$, and the result follows.
\end{proof}

With the degree invariant, conditional distribution of signs, and balance estimate established, we now utilize them in the general accounting inequality.

\begin{proof}[Proof of Theorem~\ref{thm:noisy-upper}]
  By Lemma~\ref{lem:noisy-invariant}, the preprocessing satisfies the hypotheses of Lemma~\ref{lem:master-accounting} with threshold $\Delta_{\mathrm{lab}}$.
  Combining \eqref{eq:master-accounting} with Lemma~\ref{lem:noisy-amortized} gives
  \[
    \E[|\widehat C|]
    \le
    (2-\delta_{\Delta_{\mathrm{lab}}})\cdot k^*
    -\left(\frac23-\delta_{\Delta_{\mathrm{lab}}}\right)\cdot\E[P^*]
    +O\!
    \left(
    \Delta_{\mathrm{lab}}\cdot
    e^{-\frac{\alpha^2}{128}\cdot\varepsilon^2\cdot\Delta_{\mathrm{lab}}}
    \right)k^*.
    \tag{11}\label{eq:noisy-cover-accounting}
  \]
  Since $\nicefrac{2}{3}-\delta_{\Delta_{\mathrm{lab}}}>0$ for sufficiently small $\varepsilon$, the progress term is nonpositive and may be dropped for a uniform worst-case bound.
  The chosen $\Gamma$ satisfies $(\alpha^2 \cdot \Gamma/128)\cdot\log e>2$.
  Put $L:=\log(\nicefrac{1}{\varepsilon})$.
  Since $\Gamma$ is fixed, the definition of $\Delta_{\mathrm{lab}}$ gives $\Delta_{\mathrm{lab}}=\Theta(\nicefrac{L}{\varepsilon^2})$.
  Using this estimate and $e^{-tL}=\varepsilon^{t\cdot\log e}$ for $t>0$, we obtain
  \[
    \begin{aligned}
      \Delta_{\mathrm{lab}}\cdot
      e^{-\frac{\alpha^2}{128}\cdot\varepsilon^2\cdot\Delta_{\mathrm{lab}}}
      &=O\!\left(\frac{L}{\varepsilon^2} \cdot \varepsilon^{(\alpha^2 \cdot \Gamma/128)\cdot\log e}\right)\\
      &=O\!\left(L\cdot\varepsilon^{(\alpha^2 \cdot \Gamma/128)\cdot\log e-2}\right)\\
      &=o\!\left(\frac{\log L}{L}\right).
    \end{aligned}
    \tag{12}\label{eq:noisy-error-decay}
  \]
  Taking logarithms
  \[
    \begin{aligned}
      \log\Delta_{\mathrm{lab}}
      &=\log(\nicefrac{L}{\varepsilon^2})+O(1)=2L+O(\log L) \text{, and}\\
      \log\log\Delta_{\mathrm{lab}}
      &=\log\!\left(2L+O(\log L)\right)=\log L+\log\!\left(2+O\!\left(\frac{\log L}{L}\right)\right)=\log L+O(1).
    \end{aligned}
  \]
  Substituting in \eqref{eq:halperin_delta} yields
  \[
    \begin{aligned}
      \delta_{\Delta_{\mathrm{lab}}}
      &=(2-o(1))\frac{\log\log\Delta_{\mathrm{lab}}}{\log\Delta_{\mathrm{lab}}}\\
      &=(2-o(1))\frac{\log L+O(1)}{2L+O(\log L)}\\
      &=(1-o(1))\frac{\log\log(\nicefrac{1}{\varepsilon})}{\log(\nicefrac{1}{\varepsilon})}.
    \end{aligned}
    \tag{13}\label{eq:noisy-saving-asymptotics}
  \]
  By \eqref{eq:noisy-saving-asymptotics}, the bound in \eqref{eq:noisy-error-decay} is $o(\delta_{\Delta_{\mathrm{lab}}})$.
  Thus the additive error term in \eqref{eq:noisy-cover-accounting} is $o(\delta_{\Delta_{\mathrm{lab}}})k^*$ and cannot offset its $\delta_{\Delta_{\mathrm{lab}}}k^*$ saving.
  Combining \eqref{eq:noisy-cover-accounting} and \eqref{eq:noisy-saving-asymptotics} now proves the theorem.
\end{proof}

\section{Conclusion}

We gave randomized polynomial-time algorithms for minimum Vertex Cover under Partial Predictions and Noisy Predictions, with expected approximation guarantees below $2$ for every sufficiently small fixed advice parameter $\varepsilon>0$.
Theorems~\ref{thm:main-upper} and~\ref{thm:noisy-upper} give leading coefficients $2$ and $1$, respectively, for the improvement on the scale $\log\log(\nicefrac{1}{\varepsilon})/\log(\nicefrac{1}{\varepsilon})$ as $\varepsilon\to0$, and both algorithms always return a vertex cover.
Both analyses control the preprocessing cost relative to the number of optimum vertices it removes, preserving the leading improvement obtained by applying Halperin's approximation to the remaining bounded-degree graph.

For noisy predictions, our guarantee matches the leading asymptotic improvement implied by the analysis of Aamand, Chen, Gollapudi, Silwal, and Wu~\cite[Appendix~E, equation~(55)]{ACGSW25}, after a constant rescaling of the bias parameter.
Their model supplies a degree-$d$ vertex with $d$ independent predictions about its membership in a fixed optimum cover, whereas ours supplies one.
Our algorithm obtains this guarantee by using predictions about neighboring vertices and controlling the conditional distribution of the derived signs over successive residual graphs.

The analysis leaves several questions unresolved.
Can either asymptotic guarantee be improved, or can matching lower bounds be established in these advice models?
Can the noisy-predictions guarantee be retained under weaker independence assumptions?
Our current proof uses mutual independence to establish the conditional independence in Lemma~\ref{lem:iteration-law}; extending this analysis requires accounting for the dependence between the advice and the residual graphs selected during preprocessing.

\paragraph{AI methodologies.}

We used OpenAI Codex to assist with initial exploration of ideas, drafting and editing the exposition, and checking references.
The authors verified the correctness and originality of all content, including references.

\bibliographystyle{plain}
\begingroup
\small
\bibliography{references}
\endgroup

\end{document}